%% file: main.tex
\documentclass[11pt]{article}
\usepackage{amsmath,amsfonts,epsf}
\usepackage{amssymb}
\usepackage{graphicx}
\usepackage{grffile}
\usepackage{graphicx}
\usepackage{blkarray}
\usepackage[nosort]{cite}
\usepackage{amssymb,amsmath,amsthm,amsfonts}
\usepackage{hyperref}
\usepackage{xcolor}
\usepackage{float}
\usepackage{comment}
\usepackage[english]{babel}
\usepackage{appendix}
\usepackage{quantikz}
\usepackage{tikzit}
\usetikzlibrary{backgrounds}
\usetikzlibrary{arrows}
\usetikzlibrary{shapes,shapes.geometric,shapes.misc}
\pgfdeclarelayer{edgelayer}
\pgfdeclarelayer{nodelayer}
\pgfsetlayers{background,edgelayer,nodelayer,main}
\tikzstyle{none}=[inner sep=0mm]
\usepackage{enumitem}
\input{qutrit.tikzstyles}

\makeatletter\@addtoreset{equation}{section}\makeatother

\renewcommand{\title}[1]{\vbox{\center\LARGE{#1}}\vspace{5mm}}
\renewcommand{\author}[1]{\vbox{\center#1}\vspace{5mm}}
\newcommand{\address}[1]{\vbox{\center\em#1}}

\newcommand{\Z}{\ensuremath{{\mathbb Z}}}
\newcommand{\N}{\ensuremath{{\mathbb N}}}

\newcommand{\R}{\ensuremath{\mathbb R}}
\newcommand{\C}{\ensuremath{\mathbb C}}
\newcommand{\I}{\ensuremath{{\mathcal I}}}
\newcommand{\U}{\ensuremath{{\mathcal U}}}
\newcommand{\Pmap}{\ensuremath{{\mathcal P}}}
\newcommand{\uu}{\ensuremath{{\mathbf u}}}
\newcommand{\vv}{\ensuremath{{\mathbf v}}}
\newcommand{\M}{\ensuremath{{\mathcal M}}}

\newcommand{\G}{\ensuremath{{\mathcal G}}}

\newtheorem{theorem}{Theorem}[]
\newtheorem{corollary}{Corollary}[section]

\newtheorem{defn}[corollary]{Definition}
\newtheorem{lemma}[corollary]{Lemma}
\newtheorem{proposition}[corollary]{Proposition}
\newtheorem{remark}[corollary]{Remark}

\newtheorem{notation}[corollary]{Notation}

\newtheorem{obs}[corollary]{Observation}
\definecolor{myblue}{RGB}{0,0,255}
\definecolor{mygreen}{RGB}{0,128,0}
\begin{document}
\newcommand{\categoryname}[1]{{\sf{#1}}}
\let\astap=\aap
\let\apjlett=\apjl
\let\apjsupp=\apjs
\let\applopt=\ao
\usetikzlibrary{backgrounds}
\usetikzlibrary{arrows}
\usetikzlibrary{shapes,shapes.geometric,shapes.misc}
\pgfdeclarelayer{edgelayer}
\pgfdeclarelayer{nodelayer}
\pgfsetlayers{background,edgelayer,nodelayer,main}
\tikzstyle{none}=[inner sep=0mm]

\title{Multi-qutrit exact synthesis}

\author{Amolak Ratan Kalra$^{1,2,4}$, 
Manimugdha Saikia$^{5}$, Dinesh Valluri, Sam Winnick$^{1,3}$, Jon Yard$^{1,3,4}$}

\address{${}^1$ Institute for Quantum Computing,
University of Waterloo, Waterloo, Ontario, Canada\\

${}^2$ David R. Cheriton School of Computer Science
University of Waterloo, Waterloo, Ontario, Canada\\

${}^3$ Dept. of Combinatorics $\&$ Optimization
University of Waterloo, Waterloo, Ontario, Canada\\

${}^4$ Perimeter Institute for Theoretical Physics
Waterloo, Ontario, Canada\\

${}^5$ Department of Mathematics, University of Western Ontario, London, Ontario, Canada}

\abstract{We present an exact synthesis algorithm for qutrit unitaries in $\mathcal{U}_{3^n}(\mathbb{Z}[1/3,e^{2\pi i/3}])$ over the Clifford$+T$ gate set with at most one ancilla. This extends the already known result of qutrit metaplectic gates being a subset of Clifford$+T$ gate set with one  ancilla. As an intermediary step, we construct an algorithm to convert 3-level unitaries into multiply-controlled gates, analogous to Gray codes converting 2-level unitaries into multiply-controlled gates. Finally, using catalytic embeddings, we present an algorithm to exactly synthesize unitaries $\mathcal{U}_{3^n}(\mathbb{Z}[1/3,e^{2\pi i/9}])$ over the Clifford$+T$ gate set with at most 2 ancillas. This, in particular, gives an exact synthesis algorithm of single-qutrit Clifford$+\mathcal{D}$ over the multi-qutrit Clifford$+T$ gate set with at most two ancillas.}
\\

\tableofcontents

\input{contents}

\section{Acknowledgements}
The pictures were made using the software TikZit.
ARK would like to thank NTT research for financial and technical support. The work of Jon Yard and Sam Winnick was supported in part by the NSERC Discovery under Grant No. RGPIN-2018-04742, the NSERC project FoQaCiA under Grant No. ALLRP-569582-21, and the Perimeter Institute for Theoretical Physics. 

Research at Perimeter Institute and IQC is supported by the Government of Canada through Innovation, Science and Economic Development Canada, and by the Province of Ontario through the Ministry of Research, Innovation and Science.
\bibliography{qutrit}
\bibliographystyle{ssg}
\end{document}

%% file: qutrit.tikzstyles
\tikzstyle{new style 0}=[fill=black, draw=black, shape=circle]
\tikzstyle{new style 1}=[fill=white, draw=black, shape=circle]
\tikzstyle{new style 2}=[fill=white, draw=black, shape=rectangle]
\tikzstyle{new edge style 5}=[draw=red]
\tikzstyle{new edge style 0}=[fill=none, draw=red]

\tikzstyle{new edge style 1}=[-]
\tikzstyle{new edge style 2}=[->, draw=red]
\tikzstyle{new edge style 3}=[<-, draw=red]
\tikzstyle{new edge style 4}=[->]
\tikzstyle{new edge style 0}=[draw=red, ->]

%% file: contents.tex
\section{Introduction}

Let $\G$ be a group generated by a finite set $S \subset \U_{d}(\mathbb{C})$ of $d \times d$ complex unitary matrices such that $\G$ is topologically dense in $\U_{d}(\mathbb{C}),$ i.e., any element of $\U_{d}(\mathbb{C})$ can be arbitrarily closely approximated by an element of $\G$. Such a finite set of generators is called a universal gate set. The problem of exact synthesis in quantum computing asks for an algorithm to solve the word problem in the group $\G$ in terms of a universal gate set. Several well known universal gate sets such as Clifford$+T$, Clifford$+R$, Clifford$+V$ etc. are defined over number fields, often a field of cyclotomic numbers. The problem naturally has an arithmetic flavour.

The exact synthesis algorithms for $d=2,$ i.e., for qubits over the Clifford$+T$, Clifford$+V$ and other similar gate sets use methods from number theory such as quaternion factorization\cite{kmm,giles2013exact,yardapprox,kliuchnikov2015framework} and more recently catalytic embedding \cite{amy2023catalytic}. The main result in the single-qubit case for the Clifford$+T$ case was shown in \cite{kmm}. It states that the set of single qubit unitary matrices exactly implementable over the Clifford$+T$ gate set are precisely $\U_2(\mathbb{Z}[\frac{1}{\sqrt{2}}, i])$ the unitary matrices whose entries are in the ring $\mathbb{Z}[\frac{1}{\sqrt{2}},i]$. Moreover, the proof of this yields an efficient algorithm to solve the word problem for Clifford$+T$. This was later extended to multi-qubits by Giles and Selinger \cite{giles2013exact}. 

For $d= 3$, a synthesis algorithm is already known for single qutrit Clifford$+R$ gates as studied in \cite{vadymanyons} and later in \cite{kalra2023synthesis}. More precisely, the group $\U_{3}(\mathbb{Z}[\frac{1}{3}, e^{2\pi i/3}])$ is generated by the Clifford gates and $R = \mathrm{diag}(1,1,-1)$ gate. Moreover, given an arbitrary matrix $U$ in  $\U_{3}(\mathbb{Z}[\frac{1}{3}, e^{2\pi i/3}])$ there is an efficient algorithm to produce a sequence of gates from the Clifford and $R$ gates whose product is $U$. Such an algorithm is unknown for other more desirable variants such as the single-qutrit Clifford$+T$ gates. The Clifford$+T$ gates form a proper subgroup of $\U_{3}(\mathbb{Z}[\frac{1}{3}, e^{\frac{2\pi i}{9}}]).$ It was recently conjectured in \cite{kalra2023synthesis} and proved in  \cite{evra2024arithmeticity} that $\U_{3}(\mathbb{Z}[\frac{1}{3}, e^{2\pi i/9}])$ is generated by Clifford$+\mathcal{D}$ gates. This naturally asks for an extension of synthesis to the multi-qutrit gates $\U_{3^n}(\mathbb{Z}[\frac{1}{3}, e^{2\pi i/9}]).$ In this paper, we answer this question by means of a catalytic embedding of $\U_{3^n}(\mathbb{Z}[\frac{1}{3}, e^{\frac{2\pi i}{9}}])$ in $\U_{3^{n+1}}(\mathbb{Z}[\frac{1}{3}, e^{\frac{2\pi i}{3}}])$ and by developing synthesis in the group $\U_{3^{n}}(\mathbb{Z}[\frac{1}{3}, e^{\frac{2\pi i}{3}}])$ in terms of multi-qutrit Clifford$+T$ synthesis.

One interesting aspect of this approach is that the catalytic embedding allows us to bypass the need for using the more subtle single qutrit Clifford $+\mathcal{D}$ synthesis and instead use a more simpler  synthesis algorithm in the group $\U_{3^{n}}(\mathbb{Z}[\frac{1}{3}, e^{\frac{2\pi i}{3}}])$. The method of proof for this group relies on the `derivative mod $3$' trick from \cite{kalra2023synthesis}, as it was the case for the single-qutrit version. 

In Section \ref{prelim}, we recall the necessary qutrit gate sets, the notion of catalytic embeddings and some basic facts about the (localized) ring of  cyclotomic integers $\mathbb{Z}[\frac{1}{p},\zeta_{p^{l}}]$ which plays a crucial role in rest of the paper.

In Section \ref{sec3}, we give an algorithm for multi-qutrit synthesis of the gates in $\U_{3^n}(\mathbb{Z}[\frac{1}{3},e^{\frac{2\pi i}{3}}])$. We follow Giles and Selinger to first prove that any matrix $U \in \U_{3^n}(\mathbb{Z}[\frac{1}{3},e^{\frac{2\pi i}{3}}])$ can be written as a product of $3$-level matrices
of `type' $H,S,X$ and $R$ using the single-qutrit Clifford $+R$ algorithm. These are further converted into multiply controlled gates of type $H,S,X$ and $R$ respectively using an analog of Gray codes for qutrits. Each such multiply controlled gate is represented by an explicit Clifford $+T$ circuit with at most one ancilla using results in \cite{Yeh_2022}.  

In Section \ref{sec4} we use the results in Section \ref{sec3} and an explicit catalytic embedding $\Phi :  \U_{3^n}(\Z[\frac{1}{3},\zeta]) \hookrightarrow \U_{3^{n+1}}(\Z[\frac{1}{3},\omega])$  to prove that any gate in $\U_{3^n}(\mathbb{Z}[\frac{1}{3}, e^{\frac{2\pi i}{9}}])$ can be synthesized using Clifford$+T$ gates with at most two ancillae. This is achieved by providing an embedding from $\U_{3^n}(\mathbb{Z}[\frac{1}{3}, e^{\frac{2\pi i}{3}}])$ into the group of $(n+2)$-qutrit Clifford $+T$ gates obtained by composing the Catalytic embedding and the embedding $\U_{3^{n+1}}(\Z[\frac{1}{3},\omega]) \hookrightarrow (\text{Clifford}+T)^{(n+2)}$ from Section \ref{sec3}. \\
\textbf{Acknowledgement}: After the present work was completed, it was brought to our attention that Andrew N. Glaudell, Neil J. Ross, John van de Wetering and Lia Yeh independently established related results in \cite{julien}.


\section{Preliminaries}\label{prelim}

In this section we introduce three gate sets, which are the primary focus of this paper. To address the question of exact synthesis over these gate sets, we recall some properties of cyclotomic ring extensions and the recently introduced catalytic embedding methods discussed in \cite{amy2023catalytic}. \\


\subsection{Qutrit Gate Sets}
\begin{defn}
The single-qutrit Clifford$+\mathcal{D}$ group is defined as the group generated by the following matrices:
\[
H=\frac{1}{\sqrt{-3}}
\begin{bmatrix}
1 & 1 & 1\\
1 &  \omega & \omega^{2}\\
1 & \omega^{2} & \omega
\end{bmatrix}~~~
\mathcal{D}_{[a,b,c]}=\begin{bmatrix}
\zeta^{a} & 0 & 0\\
0 & \zeta^{b} & 0\\
0 & 0 & \zeta^{c}
\end{bmatrix}~~~
R_{[a,b,c]}=
\begin{bmatrix}
(-1)^a & 0 & 0\\
0 & (-1)^b & 0\\
0 & 0 & (-1)^c
\end{bmatrix}.
\]
\end{defn}
\begin{defn}
The single-qutrit Clifford$+T$ gate set is defined as the group generated by the following matrices:
\[
H=\frac{1}{\sqrt{-3}}
\begin{bmatrix}
1 & 1 & 1\\
1 &  \omega & \omega^{2}\\
1 & \omega^{2} & \omega
\end{bmatrix}~~~
S=\begin{bmatrix}
1 & 0 & 0\\
0 & 1 & 0\\
0 & 0 & \omega
\end{bmatrix}~~~~
T=\begin{bmatrix}
1 & 0 & 0\\
0 & \zeta & 0\\
0 & 0 & \zeta^{-1}
\end{bmatrix}.
\]

\end{defn}
\begin{defn}
The single-qutrit Clifford$+R$ gate set is defined as the group generated by the following matrices: 
\[
H=\frac{1}{\sqrt{-3}}
\begin{bmatrix}
1 & 1 & 1\\
1 &  \omega & \omega^{2}\\
1 & \omega^{2} & \omega
\end{bmatrix}~~~
S=\begin{bmatrix}
1 & 0 & 0\\
0 & 1 & 0\\
0 & 0 & \omega
\end{bmatrix}~~~~
R=\begin{bmatrix}
1 & 0 & 0\\
0 & 1 & 0\\
0 & 0 & -1
\end{bmatrix}
\]
\end{defn}
Note that in all the definitions above $\zeta=e^{\frac{2\pi i}{9}}$ and $\omega=e^{\frac{2\pi i}{3}}$
One way to extend these to universal multi-qutrit gate sets \cite{Nielsen_Chuang_2010} is to adjoin the $CX$ gate defined below:
\[
CX(\ket{x}\ket{y})=\ket{x}\ket{x+y~\text{(mod}~3)}.
\]
The multi-qubit Clifford$+\mathcal D$, Clifford$+T$, and Clifford$+R$ gate sets are obtained by adjoining $CX$ to their single qutrit counterparts. Further, note that the Clifford$+T$ and Clifford$+R$ gate sets are special cases of the Clifford+$\mathcal{D}$ gate set.

\begin{notation}[Single-qutrit permutation gates]
The permutation matrix corresponding to $\sigma\in S_3$ is called a permutation gate and is denoted by $X_{\sigma}$. When $\sigma = (1\;2\;3)$, then we get the standard $X$ gate for qutrit, that is
$$X = X_{(1\;2\;3)} = 
\begin{pmatrix}
0 & 0 & 1 \\ 1 & 0 & 0\\ 0 & 1 & 0\\
\end{pmatrix}.$$   
\end{notation}
\begin{remark} \label{rem:H,X_gen_S_3}
Note that $H^2 = X_{(2\;3)}$ and as we know a 2-cycle and a 3-cycle generate the whole group $S_3$, so $X$ and $H^2$ generate all the permutation gates. 
\end{remark}


\subsection{Catalytic Embeddings}\label{seccat}

In this section, we construct a catalytic embedding, in the sense of \cite{amy2023catalytic}, which will be used in the proof of Theorem \ref{th2}.
\begin{defn}[Catalytic embedding]
Let $\G_1$ be a subgroup of $\U_n(\C)$ and $\G_2$ be a subgroup of $\U_{nd}(\C)$. We call a group homomorphism $\Phi: \G_1 \to \G_2$ a $d$-dimensional catalytic embedding with respect to a quantum state $\ket{\lambda} \in \C^d$ if
$$\Phi(U)(\ket{u}\otimes \ket{\lambda}) = (U\ket{u}) \otimes \ket{\lambda} \text{ \quad for all } \ket{u} \in \C^n$$
\end{defn}

Let $R=\mathbb Z[\frac1d]$. We explicitly describe a catalytic embedding of $\U_n(R[\zeta_{d^k}])$ into $ \U_{nd}(R[\zeta_{d^{k-1}}])$. Let $U \in \U_n(R[\zeta_{d^k}])$. The minimal polynomial of $\zeta_{d^k}$ over $R[\zeta_{d^{k-1}}]$ is $x^d-\zeta_{d^{k-1}}$. It follows that $1, \zeta_{d^k}, \ldots, \zeta_{d^{k}}^{d-1}$ are linearly independent over $R[\mathbb{\zeta}_{d^{k-1}}]$  and that they span $R[\zeta_{d^k}]$.  In other words, an arbitrary element $u \in R[\zeta_{d^k}]$ can be written uniquely as $\sum_{j=0}^{d-1} a_j \zeta_{d^{k}}^j,$ where $a_j \in R[\zeta_{d^{k-1}}].$
    Therefore, there exist unique $A_1, A_2, ..., A_d \in \mathcal{M}_n(R[\zeta_{d^{k-1}}])$ such that: $$U= \sum_{j=0}^{d-1}\zeta_{d^k}^jA_j.$$
    We define $\Phi_k: \U_n(R[\zeta_{d^k}]) \to \U_{nd}(R[\zeta_{d^{k-1}}])$ by: $$\Phi_k(U) = \sum_{j=0}^{d-1} A_j \otimes \Lambda_k^j,$$
    where $\Lambda_k = \left(
    \begin{array}{c|c}
    0 & \zeta_{d^{k-1}} \\
    \hline
    I_{d-1} & 0 \\
    \end{array}
    \right)$. It is easy to verify that $\Phi_k$ satisfies the requirements of a catalytic embedding with $\ket{\lambda_k} = (\zeta_{d^k}^{d-1},\zeta_{d^k}^{d-2},...,1)^t$ as defined above.

In particular, for each $n\in \N$ we have a catalytic embedding $\U_{3^n}(\Z[\frac{1}{3}, \zeta]) \hookrightarrow \U_{3^{n+1}}(\Z[\frac{1}{3},\omega])$.

\subsection{Facts from Number Theory }
\begin{proposition} \label{prop:1/p1/chi}
Let $p$ be a prime number. For $l\in \N$, let $\zeta_{p^l}$ be a primitive $p^l$-th root of unity and $\chi_{p^l} = 1 -\zeta_{p^l}$. Then $\Z\left[\frac{1}{p},\zeta_{p^l}\right] = \Z\left[\frac{1}{\chi_{p^l}},\zeta_{p^l}
\right]$.
\end{proposition}
\begin{proof}
It is well known that the value of the $p^l$-th cyclotomic polynomial at $1$ is $p$, so
    $$p= \prod_{\substack{1\le k\le p^l \\
    \gcd(k,p)=1}} (1- \zeta_{p^l}^k) = (1-\zeta_{p^l})^{\phi(p^{l})} \prod_{\substack{1\le k\le p^l \\
    \gcd(k,p)=1}} \frac{1- \zeta_{p^l}^k}{1-\zeta_{p^l}}.$$
    We know that for all $k$ with $\gcd(k,p)=1$, the term $\frac{1- \zeta_{p^l}^k}{1-\zeta_{p^l}}$ is a unit in $\Z[\zeta_{p^l}]$. Therefore $p = \chi_{p^l}^{\phi(p^{l})} u $ for some unit $u \in\Z[\zeta_{p^l}]$. Hence we can write
    $$\frac{1}{p} =  \frac{u^{-1}}{\chi_{p^l}^{\phi(p^{l})}} \in \Z\left[\frac{1}{\chi_{p^l}},\zeta_{p^l}
    \right] \text{ and } \frac{1}{\chi_{p^l}} = \frac{(1-\zeta_{p^l})^{(\phi(p^{l})-1)} u}{p} \in  \Z\left[\frac{1}{p},\zeta_{p^l}\right]$$
\end{proof}

\begin{notation}
From now on when the subscript is not specified, then by $\zeta$ we mean $e^{2\pi i/9} $, a primitive 9th root of unity. Also we denote $\omega=e^{2\pi i/3}$, a primitive 3rd root of unity and $\chi = 1-\omega$.
\end{notation}
As a consequence of Proposition \ref{prop:1/p1/chi}, we see that localizing the element 3 or localizing the element $\chi$ produces the same ring $\Z[\frac{1}{3},\omega]$. Further note that $\frac{1}{\sqrt{-3}} = \frac{\omega^2}{\chi}$ and $\omega^2$ is a unit, therefore localizing at ${\sqrt{-3}}$ also produces the same ring i.e:
$$\Z\left[\frac{1}{3},\omega\right] = \Z\left[\frac{1}{\sqrt{-3}},\omega\right]=\Z\left[\frac{1}{\chi},\omega\right].$$
 We can talk about denominator exponent and smallest denominator exponents of an element $x \in \Z[\frac{1}{3},\omega]$ with respect to $\chi$. We formally define them below:

\begin{defn}[sde of an element w.r.t. $\chi$]
Let $x\in \Z[\frac{1}{3},\omega]$, then $f \in \N \cup \{0\}$ is said to be a \textit{denominator exponent} of $x$ with respect to $\chi$ if $\chi^f x \in \Z[\omega]$. Further if $\chi^f x \in \Z[\omega]$ but $ \chi^{f-1} x \notin \Z[\omega]$, then $f$ is said to be the \textit{smallest denominator exponent} or \textit{sde} of $x$ (with respect to $\chi$). 
\end{defn}
\begin{defn}[sde of a column]
Let $\uu\in \Z[\frac{1}{3},\omega]^m$ be a $m\times 1$ column, then $f \in \N \cup \{0\}$ is said to be a \textit{denominator exponent} of $\uu$ with respect to $\chi$ if $\chi^f \uu \in \Z[\omega]^m$. Further if $\chi^f \uu \in \Z[\omega]^m$ but $ \chi^{f-1} \uu \notin \Z[\omega]^m$, then $f$ is said to be the \textit{smallest denominator exponent} or \textit{sde} of $\uu$ (with respect to $\chi$).
\end{defn}
\section{Unitaries with entries from $\mathbb{Z}{ [\frac{1}{3},\omega]}$} \label{sec3}

 Given a $n$-qutrit gate $U$ with entries in the ring $\Z[\frac{1}{3},\omega]$ we want to make a circuit equivalent to $U$ using multi-qutrit Clifford$+T$ gate sets. We already know that single-qutrit Clifford$+T$ cannot generate every unitary gate with entries in $\Z[\frac{1}{3},\omega]$, in particular, the $R$ gate \footnote{One can prove this by observing that the $R$ gate has sde $=0$ and by characterizing the subgroup of Clifford $+T$ matrices of sde $=0$. Such a characterization can be done by solving certain integral quadratic forms, see Section 4 of \cite{kalra2023synthesis}.
 }. But in \cite[Corollary 23]{metaplectic}, it was shown that a single-qutrit Clifford$+R$ gate can be represented using a circuit over Clifford$+T$ given an extra ancilla. Our Theorem \ref{th1} extends \cite[Corollary 23]{metaplectic} to make an $n$-qutrit gate with entries from $\Z[\frac{1}{3},\omega]$ in terms of a Clifford$+T$ circuit with an extra ancilla.

\begin{theorem} \label{generalizationGS} \label{th1}
Let $U \in \U_{3^n}(\Z[\frac{1}{3},\omega])$. Then $U$ can be exactly represented by a quantum circuit over Clifford$+T$, using at most one ancilla.
More precisely, there exists an injective group homomorphism 
$$\U_{3^n}(\mathbb{Z}[\frac{1}{3}, e^{\frac{2\pi i}{3}}]) \xhookrightarrow{\Psi} (\text{Cliff}+T)^{(n+1)}$$
and nonzero $\ket w$ such that 
$\Psi(U)\ket{v}\otimes \ket{w} = U \ket{v} \otimes \ket{w}$

\end{theorem}

Note that although the extra ancilla seems undisturbed by the embedding $\Psi$, the circuit representation of an arbitrary gate often involves a non-trivial action on the anicilla.

The proof of the above theorem goes through a series of steps. The first step is to express $U \in \U_{3^n}(\Z[\frac{1}{3},\omega])$ as a product of 3-level unitaries, which are then converted into (multiply)-controlled single qutrit Clifford$+R$ gates and finally to a circuit over multi-qutrit Clifford$+T$. Below we explain these steps in detail by defining necessary tools for the proof of Theorem \ref{generalizationGS}.

\subsection{Reduction to $3$-level unitaries}

\begin{defn}[3-level matrices]
Let $R$ be a ring. Let $U \in \M_3(R)$ be given by the matrix:
\[
U=\begin{pmatrix}
a_1 & a_2 & a_3\\
b_1 & b_2 & b_3\\
c_1 & c_2 & c_3
\end{pmatrix}
\]
A matrix $M \in \M_m(R)$ (with $m\ge 3$) is said to be a 3-level matrix of type $U$ if there is a coordinate subspace of dimension $3$, say with indices $j_1$, $j_2$, and $j_3$, on which $M$ acts by $U$, and such that $M$ acts trivially on its orthogonal complement. In other words:
\[
M=\begin{blockarray}{cccccccc}
    &  & j_1 & & j_2 &  & j_3 &  \\
    \begin{block}{c(c|c|c|c|c|c|c)}
       & I &&&&&& \\\cline{2-8}
      j_1 & & a_1 && a_2 && a_3 \\\cline{2-8}
       & && I &&&& \\\cline{2-8}
      j_2 & & b_1 && b_2 && b_3  \\\cline{2-8}
       &&&&& I && \\ \cline{2-8}
      j_3 & & c_1 && c_2 && c_3 \\\cline{2-8}
       &&&&&&& I\\ 
    \end{block}
  \end{blockarray}
\]
We denote the above matrix as $U_{[j_1,j_2,j_3]}$.
\end{defn}

\begin{defn}[1 and 2-level matrices]
We define 1-level and 2-level matrices similarly to the above definition. For $U=\begin{pmatrix}
    a & b\\ c & d
\end{pmatrix} \in \M_2$, a 1-level matrix of type $a$ and 2-level matrices $U$ has the form:
\begin{align*}
    [a]_{[j]} = \begin{blockarray}{cccc}
    && j &\\
    \begin{block}{c(c|c|c)}
     & I && \\\cline{2-4}
      j & & a & \\\cline{2-4}
     & && I\\
    \end{block}
  \end{blockarray}  \;, \;\;\;\; & \;\;\;\; 
   U_{[j,l]}=\begin{blockarray}{cccccc}
    && j & & l& \\
    \begin{block}{c(c|c|c|c|c)}
       & I &&&& \\\cline{2-6}
      j & & a && b \\\cline{2-6}
      & && I && \\\cline{2-6}
      l & & c && d &  \\\cline{2-6}
      &&&&& I \\
    \end{block}
  \end{blockarray}
\end{align*}
respectively.
\end{defn}
\begin{remark}
    Note that by definition, if the size of the matrix is greater than or equal to 3, then all 1-level and 2-level are also 3-level matrices. In particular, if the size of the matrix is greater than or equal to 3, then a 3-level unitary of type $S$ or $R$ is equivalent to a 1-level unitary of type $\omega$ or $-1$  respectively.
\end{remark}

We use the map $\Pmap: \Z[\omega] \to \Z/3\Z$ defined by $g(\omega) \mapsto g(1)$. Recall that $g \in\chi\Z[\omega]$ iff $\Pmap(g) = 0$, and additionally, $g \in\chi^2 \Z[\omega]$ iff $\Pmap(g) = g(1) = 0 $ and $\Pmap(g') = g'(1) = 0$ where $g'(\omega)$ is the formal derivative of $g$ with respect to the formal variable $\omega$. \cite{kalra2023synthesis}. 

\begin{lemma} \label{sde0}
Let $\mathbf{u} \in \Z[\frac{1}{3},\omega]^m$ be an $m$-dimensional column vector of norm 1 and sde 0. Then $\uu$ is some permutation of $(\pm\omega^a,0,....,0)^t$ for some $a\in\mathbb Z$.
\end{lemma}
\begin{proof}
As sde$(\uu)=0$, so $\uu = (u_1,...,u_m) \in \Z[\omega]^m$. Let $u_j = a_j \omega +b_j $ for some $a_j,b_j \in \Z$. Using the fact that the norm of $\uu$ is 1, we get the equation $$\sum_{j=1}^m a_j^2 + b_j^2 -a_jb_j = 1.$$ Clearly, up to permutation the only solutions for the above equation are i) $a_1 = \pm 1$ and all the other variables are 0, ii) $b_1=\pm 1$ and all the other variables are 0, and iii) $(a_1,b_1)=(1,1)$ or $(-1,-1)$ and all other variables are 0. In all these cases, we see that $u_1 =\pm\omega^a$ for some $a \in \mathbb Z$. This proves our statement.
\end{proof}

\begin{lemma} \label{lemma:m=1,2}
    Let $m=1$ or $2$. If $\uu\in\mathbb Z[\frac{1}{3},\omega]^m$ is a unit vector then $\mathrm{sde}(\uu)=0$.
\end{lemma}
\begin{proof}
    To the contrary, suppose $\mathrm{sde}(\uu)=f>0$ and $\uu$ is a unit vector. Then we have $\chi^f\uu = \vv$ for some $\vv = (v_1,...,v_m)^t\in \Z[\omega]^m$ but $\chi^{f-1}\uu \notin \Z[\omega]^m$. Applying $\Pmap$ to both sides of $\sum_{j=1}^m |v_j|^2 = |\chi|^{2f}$, we obtain:
    
    $$\sum_{j=1}^m \Pmap(v_j)^2 = 0  .$$
    Since $\Pmap(v_j)^2$ can be either 0 or 1 and $m<3$, so we get that $\Pmap(v_j)$ must be 0 for all $j$, i.e. ${v_j} \in {\chi}\Z[\omega]$, i.e. ${\vv} \in {\chi}\Z[\omega]^m$. Then $\chi^{f-1}\uu = {\vv} \in {\chi}\Z[\omega]^m$, a contradiction.
\end{proof}

\begin{corollary} \label{corollary:m=1,2}
    Let $m=1$ or 2 and $\uu \in \Z[\frac{1}{3},\omega]^m$ be a column vector of norm 1. Then there exists a sequence $U_1,...,U_k$ with each $U_j$ either a 2-level unitary of type $\begin{pmatrix}
        0 & 1\\ 1 & 0
    \end{pmatrix} $ of a 1-level unitary of type $\omega$ and $-1$ such that $U_1...U_k\uu = (1)^t$ or $(1,0)^t$ for $m=1$ and $m=2$ respectively.
\end{corollary}

\begin{proof}
    By Lemma \ref{lemma:m=1,2} we get that sde$(\uu)=0$. Using Lemma \ref{sde0}, we get that $\uu = (\pm\omega^a)$ or $\uu = (\pm\omega^a,0)$ or $(0, \pm\omega^a)$ respectively for $m=1$ and $2$. Therefore, we can use a sequence of unitaries of the form $\begin{pmatrix}
        0 & 1\\ 1 & 0
    \end{pmatrix} $ or 1-level unitaries of type $\omega$ and $-1$ to convert $\uu$ into $(1,0)$ or $(1)$ respectively for $m=1$ and 2.
\end{proof}

\begin{lemma} \label{lemma:inner_induction}
    Let $m\ge 3$ and let $\uu \in \Z[\frac{1}{3},\omega]^m$ be a unit vector with sde $f>0$. Then there exists a sequence $U_1,...,U_k$ with each $U_j$ either a 3-level unitary of type $X, H$ and $S$ or a 1-level unitary of type $-1$ such that the sde of all the entries of the resulting column vector $U_1,...,U_k\uu$ is at most $f-1$.
\end{lemma}

\begin{proof}
    We have $\chi^f\uu = \vv$ for some $\vv = (v_1,...,v_m)^t  \in \Z[\omega]^m$ but $\chi^{f-1}\uu \notin \Z[\omega]^m$, i.e. there exists at least one $j$ such that $v_j \notin \chi\Z[\omega]$. Since $\|\uu\|=1$, we have $$ \sum_{j=1}^m |v_j|^2 = |\chi|^{2f}.$$ We apply the map $\Pmap$ to get that $$\sum_{j=1}^m \Pmap(v_j)^2 = 0  .$$
    Since $\Pmap(v_j)^2$ can be 0 or 1, therefore we see that the number of $j \in \{1,...,m\}$ such that $\Pmap(v_j)^2 =1$ is a multiple of $3$, i.e. the number of $j$ such that $v_j \notin \chi\Z[\omega]$ is $3l$ for some $l\ge1$. In other words, the number of entries of $\uu$ which has sde strictly greater than $f-1$ is of the form $3l$ for some $l\ge 1$. We use induction on $l$ to prove this lemma.\\
   
 \textbf{Base case} $l=1$: There exists $j_1,j_2,j_3 \in \{1,...,m\}$ (denote $\I=\{j_1,j_2,j_3\}$) such that $\Pmap(v_j)^2=1$ for all $j \in \I$ and $\Pmap(v_j)^2 =0$ for all $j \notin \I$ (i.e. $\frac{v_j}{\chi} \in \Z[\omega]$ for all $j \notin \I$). First notice that $\Pmap(v_j) = 1 $ or $-1$ for all $j \in \I$. If $\Pmap(v_{j_i})=-1$ for some $j_i \in \I$, we can use a one level unitary $[-1]_{[j_i]}$ so that $\Pmap(v_{j_i})=1$. Therefore after the application of the above 3-level matrices, we have $\Pmap(v_j)=1$ for all $j \in \I$. Let $\Tilde{\vv} $ be the column of length 3 with entries $v_{j_1}, v_{j_2}$ and $v_{j_3}$ in its 1st, 2nd and 3rd entries respectively. Let $\Tilde{\uu} = \frac{\Tilde{\vv}}{\chi^f}$, then 
        $$HS^a\Tilde{\uu} =   \frac{\omega^2}{\chi^{f+1}} \begin{pmatrix}
        v_{j_1}+   v_{j_2} +   \omega^av_{j_3} \\
        v_{j_1}+  \omega v_{j_2} +   \omega^{a+2} v_{j_3}\\
        v_{j_1}+  \omega^2 v_{j_2} +   \omega^{a+1} v_{j_3}
        \end{pmatrix}$$
        Let $g_1 = v_{j_1}+   v_{j_2} +   \omega^av_{j_3}$, $g_2 = v_{j_1}+  \omega v_{j_2} +   \omega^{a+2} v_{j_3}$ and $v_{j_1}+  \omega^2 v_{j_2} +   \omega^{a+1} v_{j_3} $ in $\Z[\omega]$. Then we see that $\Pmap(g_1)= \Pmap(g_2) = \Pmap(g_3) = 0  $.  We will choose $a \in \Z/3\Z$ such that $\frac{g_j}{\chi^2} \in \Z[\omega]$ for each $j\in \{1,2,3\}$. We have
        \begin{align*}
            g_1'(\omega) & = v'_{j_1}+ v_{j_2}' + a\omega^{a-1}v_{j_3} + \omega^av_{j_3}'\\
            \Pmap(g_1') & = a+ \sum_{j\in \mathcal{I} } \Pmap(v_{j}') \\
            g_2'(\omega) & = v'_{j_1}+ v_{j_2}+ \omega v'_{j_2}+ (a+2)\omega^{a+1}v_{j_3} + \omega^{a+2}v_{j_3}' \\ 
            \Pmap(g_2') & = a+ 3 +\sum_{j\in \mathcal{I} } \Pmap(v_{j}') =\Pmap(g_1')   \\
            g_3'(\omega) & = v'_{j_1}+  2\omega v_{j_2}+ \omega^2 v_{j_2}' + (a+1)\omega^{a}v_{j_3} + \omega^{a+1}v_{j_3}'\\ 
            \Pmap(g_3') & = a+ 3 +\sum_{j\in \mathcal{I} } \Pmap(v_{j}') = \Pmap(g_1')  
        \end{align*}
        We choose $a \in \Z/3\Z$ such that $a+ \sum_{j\in \mathcal{I} } \Pmap(v_{j}') = 0$, then $\Pmap(g_j') = 0$ for $j \in \{1,2,3\}$. Combining with $\Pmap(g_j) = 0$, we see that $\frac{g_j}{\chi^2} \in \Z[\omega]$. Therefore we have 
        $$HS^a\Tilde{\uu} =   \frac{\omega^2}{\chi^{f-1}} \begin{pmatrix}
        \frac{g_1}{\chi^2} \\
        \frac{g_2}{\chi^2}\\
        \frac{g_3}{\chi^2}
        \end{pmatrix} \text{ with } \frac{g_j}{\chi^2} \in \Z[\omega], \text{ i.e } \chi^{f-1}HS^a\Tilde{\uu} \in \Z[\omega]^m.$$
        Therefore all entries of $HS^a\Tilde{\uu}$ have sde at most $f-1$. 
        
        \textbf{Induction step:} For the induction step we pick any 3 of the indices $j_1,j_2,j_3 \in \{1,...,m\}$ such that $\Pmap(v_{j_1}), \Pmap(v_{j_2})$ and $\Pmap(v_{j_3})$ are non-zero and use the base case algorithm to lower $l$ by 1. Then from the induction hypothesis, the statement of the lemma follows.
\end{proof}

\begin{lemma} \label{column_lemma}
    Let $m\ge 3$ and $\uu \in \Z[\frac{1}{3},\omega]^m$ be an $m$-dimensional column vector of norm 1. Then there exists a sequence $U_1,...,U_k$ of 3-level unitaries of type $X,S$ and $H$; and 1-level unitaries of type $\omega$ and $-1$ such that $U_1...U_k\uu = (1,0,...,0)^t$.
\end{lemma}
\begin{proof}
    Suppose the sde of $\uu=(u_1,...,u_m)^t \in \Z[\frac{1}{3},\omega]^m$ is $f$. If $f$ is non-zero, then we can apply an appropriate 3-level unitary matrix $U_k$ on $u$ to drop the sde of $\uu$ by at least 1 
    by Lemma \ref{lemma:inner_induction}. We can inductively proceed this way to find unitaries $U_1,\ldots, U_k$ such that $U_1\ldots U_k u$ has sde $=0$. By  Lemma \ref{sde0}, this column of sde $0$ is equal to $(\pm \omega^a, 0,0,...,0)$ up to permutation. We use a 3-level unitary of type $H^2$ and  1-level unitaries of type $\omega$ and $-1$ to get $(1,0,....,0)^t$. \\

\end{proof}

\begin{lemma} \label{lemma:step1_to_a_unitary}
    Let $U \in \U_m(\Z[\frac{1}{3},\omega])$. Then there exists a sequence $U_1,...,U_k$ of 3-level unitaries of type $X, S$ and $H$; 2-level unitaries of type $\begin{pmatrix}
        0 & 1\\ 1 & 0
    \end{pmatrix} $; and 1-level unitaries of type $\omega$ and -1 such that $U_1....U_kU=I$.
\end{lemma}
\begin{proof}
    We prove this using induction on $m$. By Corollary \ref{corollary:m=1,2}, the base case is true. Now, let $\uu $ be the first column of $U$. Then using Lemma \ref{column_lemma} and Corollary \ref{corollary:m=1,2}, there exists a sequence $U_1,...,U_{k_1}$ of 3-level unitaries of type $X, S$ and $H$; 2-level unitaries of type $\begin{pmatrix}
        0 & 1\\ 1 & 0
    \end{pmatrix} $; and 1-level unitaries of type $\omega$ and $-1$ such that $$U_1...U_{k_1}\uu = (1,0,...,0)^t.$$
    Therefore we have 
    $$U_1...U_{k_1}U = \left(\begin{array}{c|c}
        1 & 0 \\
        \hline 
        0 & U'
    \end{array}\right).$$
    Using the induction hypothesis on the unitary $U'$ of size $m-1$, the statement of the lemma follows.
\end{proof}

\begin{lemma} \label{level_matrices_lemma}
    Let $U \in \U_m(\Z[\frac{1}{3},\omega])$. Then there exists a sequence $U_1,...,U_k$ of 3-level unitaries of type $X, S$ and $H$; 2-level unitaries of type $\begin{pmatrix}
        0 & 1\\ 1 & 0
    \end{pmatrix} $; and 1-level unitaries of type $\omega$ and -1 such that $U=U_1....U_k$.
\end{lemma}
\begin{proof}
A consequence of Lemma \ref{lemma:step1_to_a_unitary} after observing that $X^{-1}=X^2, S^{-1}=S^2, H^{-1}=H^3, \begin{pmatrix}
0 & 1\\ 1 & 0
\end{pmatrix}^{-1}=\begin{pmatrix}
0 & 1\\ 1 & 0
\end{pmatrix}, \omega^{-1}=\omega^2$ and $(-1)^{-1} = -1$.    
\end{proof}

\subsection{Conversion to controlled gates}

Our next task is to convert these level-unitaries to controlled gates. Let us first define what they are and observe that when the level-unitaries are of a particular form, then they already are useful controlled gates.

\begin{defn}[Controlled gates]
Let $U \in \U_3$ be a single-qutrit gate and $x_1,....,x_n \in \{0,1,2\}$. Then the $n$-qutrit gate $\ket{x_1....x_{j-1}}\otimes \ket{x_{j+1}...x_n}$-controlled-$U$ is defined to be the unitary having the following action:
\[
\ket{y_1...y_n} \mapsto \begin{cases}
\ket{x_1....x_{j-1}}\otimes (U\ket{y_j}) \otimes\ket{x_{j+1}...x_n} & \text{ when } (y_1,...,y_{j-1},y_{j+1},...,y_n)\\
& \text{\quad}\;\; =(x_1,...,x_{j-1},x_{j+1},...,x_n) \\
\ket{y_1...y_n} & \text{otherwise.}
\end{cases}
\]    
Alternatively, we also call this gate as the $(n-1)$-controlled-$U$ gate with $j$-th target wire and with controls $x_1,...,x_{j-1},x_{j+1},...,x_n$. 
\end{defn}

\begin{remark} \label{ket(2)controlled}
Note that if the control qutrit is $\ket{0}$, then conjugating the control wire with $X^2$ and if the control qutrit is $\ket{1}$, then conjugating the control wire with $X$, we can convert any controlled gate to a $\ket{2}$-controlled gate.
\end{remark}

\begin{obs} \label{obs:Gray}
Let $M \in \U_{3^n}(\Z[\frac{1}{3},\omega])$ be a 3-level unitary of type $U \in \U_3(\Z[\frac{1}{3},\omega])$. We rename the row and column indices of the matrix $M$ as a $n$-tuple $\Vec{P}=(p_1,...,p_n) \in \{0,1,2\}^n$ with lexicographic ordering.

Define $\Delta^n$ to be the $n$-dimensional grid in $\R^n$ with vertices the set of points $\{(x_1,...,x_n):x_1,...,x_n\in  \{0,1,2\}\}$ and there is an edge between two vertices if and only if the two vertices differ exactly at one coordinate and where they differ, the difference between the coordinates is exactly 1
($\Delta^2$, for example, is the grid shown in Figures \ref{fig:controlledU} \ref{fig:controlledXU}, \ref{fig:3} and \ref{fig:4}). We observe that if the three points $\overrightarrow{P_0},\overrightarrow{P_1},$ and $\overrightarrow{P_2}$ are in a straight line along an edge of $\Delta^n$ (i.e. $n-1$ of the corresponding coordinates are the same for all the three points), then $M$ is already an $(n-1)$-multiply-controlled-$(X_{\sigma}^{-1}UX_{\sigma})$ gate for some $\sigma \in S_3$. We illustrate this for $n=2$ in Figures \ref{fig:controlledU} and \ref{fig:controlledXU}.
\end{obs}
\begin{figure}[H]
  $U_{[(0,2),(1,2),(2,2)]} = \begin{blockarray}{(ccc|ccc|ccc)}
    1 & 0 & 0 & 0 & 0 & 0 & 0 & 0 & 0\\
    0 & 1 & 0 & 0 & 0 & 0 & 0 & 0 & 0\\
    0 & 0 & a_1 & 0 & 0 & a_2 & 0 & 0 & a_3\\ \cline{1-9}
    0 & 0 & 0 & 1 & 0 & 0 & 0 & 0 & 0\\
    0 & 0 & 0 & 0 & 1 & 0 & 0 & 0 & 0\\
    0 & 0 & b_1 & 0 & 0 & b_2 & 0 & 0 & b_3\\ \cline{1-9}
    0 & 0 & 0 & 0 & 0 & 0 & 1 & 0 & 0\\
    0 & 0 & 0 & 0 & 0 & 0 & 0 & 1 & 0\\
    0 & 0 & c_1 & 0 & 0 & c_2 & 0 & 0 & c_3\\ 
  \end{blockarray}$
    \centering
\begin{tikzpicture}[scale=0.45]
	\begin{pgfonlayer}{nodelayer}
		\node [style=new style 0, inner sep=0mm, minimum size=2mm] (0) at (-10, 7) {};
		\node [style=new style 0, inner sep=0mm, minimum size=2mm] (1) at (-8, 7) {};
		\node [style=new style 0, inner sep=0mm, minimum size=2mm] (2) at (-6, 7) {};
		\node [style=new style 0, inner sep=0mm, minimum size=2mm] (3) at (-10, 5) {};
		\node [style=new style 0, inner sep=0mm, minimum size=2mm] (4) at (-8, 5) {};
		\node [style=new style 0, inner sep=0mm, minimum size=2mm] (5) at (-6, 5) {};
		\node [style=new style 0, inner sep=0mm, minimum size=2mm] (6) at (-10, 3) {};
		\node [style=new style 0,inner sep=0mm, minimum size=2mm] (7) at (-8, 3) {};
		\node [style=new style 0, inner sep=0mm, minimum size=2mm] (8) at (-6, 3) {};
		\node [style=none] (9) at (-11, 7.5) {\tiny $\overrightarrow{P_{0}} = (0,2)$};
		\node [style=none] (10) at (-8, 8) {\tiny $\overrightarrow{P_{1}}=(1,2)$};
		\node [style=none] (11) at (-5.5, 7.5) {\tiny $\overrightarrow{P_{2}}=(2,2)$};
		\node [style=none] (12) at (5, 6.5) {};
		\node [style=none] (13) at (9, 6.5) {};
		\node [style=none] (14) at (5, 3.5) {};
		\node [style=none] (15) at (9, 3.5) {};
		\node [style=new style 2] (16) at (7, 6.5) {$U$};
		\node [style=new style 1, inner sep=0.5mm, minimum size=2mm] (17) at (7, 3.5) {$2$};
		\node [style=none] (18) at (-10, 7) {};
		\node [style=none] (19) at (-0.5, 9) {a) The initial 3-level unitary};
		\node [style=none] (20) at (-8, 1.75) {b) The points are already in a straight line and };
		\node [style=none] (21) at (-8, 1) {the rows are already in the correct order};
		\node [style=none] (22) at (7, 1.75) {c) The 3-level unitary is equal};
		\node [style=none] (23) at (7, 1) { to the above controlled-$U$ gate};
	\end{pgfonlayer}
	\begin{pgfonlayer}{edgelayer}
		\draw (0) to (3);
		\draw (3) to (6);
		\draw (6) to (7);
		\draw (7) to (4);
		\draw (3) to (4);
		\draw (0) to (1);
		\draw (1) to (4);
		\draw (4) to (5);
		\draw (1) to (2);
		\draw (2) to (5);
		\draw (5) to (8);
		\draw (7) to (8);
		\draw [in=90, out=-90] (16) to (17);
		\draw (14.center) to (17);
		\draw (17) to (15.center);
		\draw (12.center) to (16);
		\draw (16) to (13.center);
	\end{pgfonlayer}
\end{tikzpicture}

\caption{The 3-level unitary $U_{[(0,2),(1,2),(2,2)]}$ is already a controlled-$U$ gate}
\label{fig:controlledU}
\end{figure}

\begin{figure}[H]
    \centering
    $U_{[(1,2),(2,2),(0,2)]} = \begin{blockarray}{(ccc|ccc|ccc)}
    1 & 0 & 0 & 0 & 0 & 0 & 0 & 0 & 0\\
    0 & 1 & 0 & 0 & 0 & 0 & 0 & 0 & 0\\
    0 & 0 & c_3 & 0 & 0 & c_1 & 0 & 0 & c_2\\ \cline{1-9}
    0 & 0 & 0 & 1 & 0 & 0 & 0 & 0 & 0\\
    0 & 0 & 0 & 0 & 1 & 0 & 0 & 0 & 0\\
    0 & 0 & a_3 & 0 & 0 & a_1 & 0 & 0 & a_2\\ \cline{1-9}
    0 & 0 & 0 & 0 & 0 & 0 & 1 & 0 & 0\\
    0 & 0 & 0 & 0 & 0 & 0 & 0 & 1 & 0\\
    0 & 0 & b_3 & 0 & 0 & b_1 & 0 & 0 & b_2\\ 
  \end{blockarray}$
\begin{tikzpicture}[scale=0.4]
	\begin{pgfonlayer}{nodelayer}
		\node [style=new style 2] (0) at (-9.25, -4) {$XUX^2$};
		\node [style=new style 1, inner sep=0.5mm] (1) at (-9.25, -7) {$2$};
		\node [style=none] (2) at (-11.25, -4) {};
		\node [style=none] (3) at (-11.25, -7) {};
		\node [style=none] (4) at (-7.25, -4) {};
		\node [style=none] (5) at (-7.25, -7) {};
		\node [style=none] (6) at (-4.75, -5.5) {$=$};
		\node [style=none] (7) at (-2.25, -4) {};
		\node [style=none] (8) at (-2.25, -7) {};
		\node [style=new style 2] (9) at (3, -4) {$U$};
		\node [style=new style 2] (10) at (6, -4) {$X$};
		\node [style=none] (11) at (8, -4) {};
		\node [style=none] (12) at (8, -7) {};
		\node [style=new style 0, inner sep=0mm, minimum size=2mm] (15) at (-14, 5) {};
		\node [style=new style 0, inner sep=0mm, minimum size=2mm] (16) at (-11.5, 5) {};
		\node [style=new style 0, inner sep=0mm, minimum size=2mm] (17) at (-9, 5) {};
		\node [style=new style 0, inner sep=0mm, minimum size=2mm] (18) at (-14, 0) {};
		\node [style=new style 0, inner sep=0mm, minimum size=2mm] (19) at (-14, 2.5) {};
		\node [style=new style 0, inner sep=0mm, minimum size=2mm] (20) at (-11.5, 2.5) {};
		\node [style=new style 0, inner sep=0mm, minimum size=2mm] (21) at (-9, 2.5) {};
		\node [style=new style 0, inner sep=0mm, minimum size=2mm] (22) at (-11.5, 0) {};
		\node [style=new style 0, inner sep=0mm, minimum size=2mm] (23) at (-9, 0) {};
		\node [style=none] (24) at (6.25, 5.75) {\tiny$\overrightarrow{Q_{2}}$};
		\node [style=none] (25) at (1.25, 5.75) {\tiny$\overrightarrow{Q_{0}}$};
		\node [style=none] (26) at (3.75, 5.75) {\tiny$\overrightarrow{Q_{1}}$};
		\node [style=none] (27) at (-14, 6) {};
		\node [style=none] (28) at (-11.5, 6) {};
		\node [style=none] (29) at (-9, 6) {};
		\node [style=none] (30) at (-14, 6.5) {};
		\node [style=none] (31) at (-9, 6.5) {};
		\node [style=none] (33) at (-4, 8) {a) The initial 3-level unitary};
		\node [style=new style 0, inner sep=0mm, minimum size=2mm] (34) at (1.25, 5) {};
		\node [style=new style 0, inner sep=0mm, minimum size=2mm] (35) at (3.75, 5) {};
		\node [style=new style 0, inner sep=0mm, minimum size=2mm] (36) at (6.25, 5) {};
		\node [style=new style 0, inner sep=0mm, minimum size=2mm] (37) at (1.25, 0) {};
		\node [style=new style 0, inner sep=0mm, minimum size=2mm] (38) at (1.25, 2.5) {};
		\node [style=new style 0, inner sep=0mm, minimum size=2mm] (39) at (3.75, 2.5) {};
		\node [style=new style 0, inner sep=0mm, minimum size=2mm] (40) at (6.25, 2.5) {};
		\node [style=new style 0, inner sep=0mm, minimum size=2mm] (41) at (3.75, 0) {};
		\node [style=new style 0, inner sep=0mm, minimum size=2mm] (42) at (6.25, 0) {};
		\node [style=none] (43) at (-15, 5.5) {\tiny$\overrightarrow{P_2}=(0,2)$};
		\node [style=none] (44) at (-11.5, 5.5) {\tiny$\overrightarrow{P_0}=(1,2)$};
		\node [style=none] (45) at (-8, 5.5) {\tiny$\overrightarrow{P_1}=(2,2)$};
		\node [style=none] (46) at (-4, -1.5) {};
		\node [style=none] (47) at (-4, -1.5) {b) Permute the rows and columns to the correct order};
		\node [style=none] (48) at (-6.5, 2.5) {};
		\node [style=none] (49) at (-1.5, 2.5) {};
		\node [style=none] (50) at (-4, 3.5) {};
		\node [style=none] (51) at (-4, 3.5) {\footnotesize conjugate by};
		\node [style=none] (52) at (-2.5, 1.5) {};
		\node [style=none] (53) at (-5.5, 1.5) {};
		\node [style=none] (54) at (-5.5, 0) {};
		\node [style=none] (55) at (-2.5, 0) {};
		\node [style=new style 2,] (56) at (-4, 1.5) {$X$};
		\node [style=none] (58) at (-4, -8.5) {};
		\node [style=none] (59) at (-4, -8.5) {c) The final circuit for $U_{[(1,2),(2,2),(0,2)]}$};
		\node [style=none] (60) at (-19, 2.5) {$U_{[(1,2),(2,2),(0,2)]} \rightarrow $};
		\node [style=none] (61) at (11, 2.5) {$\leftarrow U_{[(0,2),(1,2),(2,2)]}$};
		\node [style=new style 2] (62) at (0, -4) {$X^2$};
		\node [style=new style 1, inner sep=0.5mm] (63) at (0, -7) {$2$};
		\node [style=new style 1, inner sep=0.5mm] (64) at (3, -7) {$2$};
		\node [style=new style 1, inner sep=0.5mm] (65) at (6, -7) {$2$};
		\node [style=new style 1, inner sep=0.5mm, minimum size=0mm] (66) at (-4, 0) {$2$};
	\end{pgfonlayer}
	\begin{pgfonlayer}{edgelayer}
		\draw (2.center) to (0);
		\draw (0) to (4.center);
		\draw (3.center) to (1);
		\draw (1) to (5.center);
		\draw (0) to (1);
		\draw (7.center) to (9);
		\draw (9) to (10);
		\draw (10) to (11.center);
		\draw (15) to (19);
		\draw (19) to (18);
		\draw (18) to (22);
		\draw (22) to (20);
		\draw (19) to (20);
		\draw (15) to (16);
		\draw (16) to (20);
		\draw (20) to (21);
		\draw (16) to (17);
		\draw (22) to (23);
		\draw (23) to (21);
		\draw (17) to (21);
		\draw (34) to (38);
		\draw (38) to (37);
		\draw (37) to (41);
		\draw (41) to (39);
		\draw (38) to (39);
		\draw (34) to (35);
		\draw (35) to (39);
		\draw (39) to (40);
		\draw (35) to (36);
		\draw (41) to (42);
		\draw (42) to (40);
		\draw (36) to (40);
		\draw (53.center) to (52.center);
		\draw (62) to (63);
		\draw (8.center) to (63);
		\draw (63) to (64);
		\draw (64) to (65);
		\draw (65) to (12.center);
		\draw (10) to (65);
		\draw (9) to (64);
		\draw (54.center) to (66);
		\draw (66) to (55.center);
		\draw (56) to (66);
		\draw [style=new edge style 2, bend left=15] (30.center) to (31.center);
		\draw [style=new edge style 2, bend left=345] (28.center) to (27.center);
		\draw [style=new edge style 2, bend left=345, looseness=1.25] (29.center) to (28.center);
		\draw [style=new edge style 4] (48.center) to (49.center);
	\end{pgfonlayer}
\end{tikzpicture}

 \caption{The equivalent circuit for $U_{[(1,2),(2,2),(0,2)]}$ using controlled gates}
    \label{fig:controlledXU}
\end{figure}
\begin{lemma}\label{convert_controlled_lemma}
Let $M\in \U_m(\Z[\frac{1}{3},\omega])$ be a 3-level unitary of type $U \in \U_3(\Z[\frac{1}{3},\omega])$ such that $3^{n-1}<m\le 3^n$ for some $n\ge 2$. Let $m' = 3^n-m$, then the unitary $I_{m'}\oplus M \in \U_{3^n}(\Z[\frac{1}{3},\omega])$ can be decomposed into $n$-qutrits gates consisting of $\ket{2}^{\otimes (n-1)}$-controlled-$X$, $\ket{2}^{\otimes (n-1)}$-controlled-$H$ and a controlled-$U$ gate.
\end{lemma}

\begin{proof}
Let $M' = I_{m'}\oplus M$. Then $M'$ is also a $3$-level unitary of type $U$. We rename the row and column indices of the matrix $M'$ as a $n$-tuple $\Vec{P}=(p_1,...,p_n) \in \{0,1,2\}^n$ with lexicographic ordering. Suppose $M'= U_{[\overrightarrow{P_0},\overrightarrow{P_1},\overrightarrow{P_2}]}$.\\
    
The idea is to travel along the edges of the grid $\Delta^n$ to start from $\overrightarrow{P_j}$ and end up at $\overrightarrow{P_j'}$ such that the points $\overrightarrow{P_0'}, \overrightarrow{P_1'}$ and $\overrightarrow{P_2'}$ are in a straight line along an edge of $\Delta^n$ so that we get the cases of the Observation \ref{obs:Gray}. Each of the movements of a point along the paths of $\Delta^n$ from one vertex to another vertex that changes only one of the coordinates of the point corresponds to conjugating $M'$ by an $(n-1)$-multiply-controlled permutation gate (with the target on the wire where the coordinate changes and controls on all other $n-1$ wires). We illustrate this for the case $n=2$ in Figures \ref{fig:3} and \ref{fig:4}.\\
    
    

Finally, by Remark \eqref{rem:H,X_gen_S_3} and Remark \eqref{ket(2)controlled}, the statement of the lemma follows.
\end{proof}
\begin{remark}
    The procedure in the above lemma is an analogue of the Gray code construction \cite[section 4.5.2]{Nielsen_Chuang_2010} for qutrits. Further, note that we can extend the same procedure to convert any $d$-level unitaries (similar to $3$-level unitaries) into controlled-qudit gates. In that case, we replace $\Delta^n$ by the graph $\Delta^n_d$ where     $\Delta^n_d$ is the graph in $\R^n$ with the set of vertices $\{(x_1,...,x_n):x_1,...,x_n\in  \{0,1,...,d-1\}\}$ and there is an edge between two vertices if and only if the two vertices differ exactly at one coordinate and where they differ, the difference between the coordinates is exactly 1. All other arguments used in the proof of Lemma \ref{convert_controlled_lemma} can also be mimicked for qudits.
\end{remark}

\begin{figure}

    \centering
        $U_{[(0,1),(1,2),(2,2)]} = \begin{blockarray}{(ccc|ccc|ccc)}
    1 & 0 & 0 & 0 & 0 & 0 & 0 & 0 & 0\\
    0 & a_1 & 0 & 0 & 0 & a_2 & 0 & 0 & a_3\\
    0 & 0 & 1 & 0 & 0 & 0 & 0 & 0 & 0\\ \cline{1-9}
    0 & 0 & 0 & 1 & 0 & 0 & 0 & 0 & 0\\
    0 & 0 & 0 & 0 & 1 & 0 & 0 & 0 & 0\\
    0 & b_1 & 0 & 0 & 0 & b_2 & 0 & 0 & b_3\\ \cline{1-9}
    0 & 0 & 0 & 0 & 0 & 0 & 1 & 0 & 0\\
    0 & 0 & 0 & 0 & 0 & 0 & 0 & 1 & 0\\
    0 & c_1 & 0 & 0 & 0 & c_2 & 0 & 0 & c_3\\
  \end{blockarray}$
\begin{tikzpicture}[scale=0.45]
	\begin{pgfonlayer}{nodelayer}
		\node [style=new style 0, inner sep=0mm, minimum size=2mm] (11) at (-6, 6) {};
		\node [style=new style 0, inner sep=0mm, minimum size=2mm] (12) at (-4, 6) {};
		\node [style=new style 0, inner sep=0mm, minimum size=2mm] (13) at (-2, 6) {};
		\node [style=new style 0, inner sep=0mm, minimum size=2mm] (14) at (-6, 4) {};
		\node [style=new style 0, inner sep=0mm, minimum size=2mm] (15) at (-4, 4) {};
		\node [style=new style 0, inner sep=0mm, minimum size=2mm] (16) at (-2, 4) {};
		\node [style=new style 0, inner sep=0mm, minimum size=2mm] (17) at (-6, 2) {};
		\node [style=new style 0, inner sep=0mm, minimum size=2mm] (18) at (-4, 2) {};
		\node [style=new style 0, inner sep=0mm, minimum size=2mm] (19) at (-2, 2) {};
		\node [style=new style 0, inner sep=0mm, minimum size=2mm] (22) at (8, 6) {};
		\node [style=new style 0, inner sep=0mm, minimum size=2mm] (23) at (10, 6) {};
		\node [style=new style 0, inner sep=0mm, minimum size=2mm] (24) at (12, 6) {};
		\node [style=new style 0, inner sep=0mm, minimum size=2mm] (25) at (8, 4) {};
		\node [style=new style 0, inner sep=0mm, minimum size=2mm] (26) at (10, 4) {};
		\node [style=new style 0, inner sep=0mm, minimum size=2mm] (27) at (12, 4) {};
		\node [style=new style 0, inner sep=0mm, minimum size=2mm] (28) at (8, 2) {};
		\node [style=new style 0, inner sep=0mm, minimum size=2mm] (29) at (10, 2) {};
		\node [style=new style 0, inner sep=0mm, minimum size=2mm] (30) at (12, 2) {};
		\node [style=none] (31) at (-2, -2) {};
		\node [style=none] (32) at (-2, -5) {};
		\node [style=new style 1, inner sep=1mm, minimum size=0mm] (33) at (0, -2) {$0$};
		\node [style=new style 2] (34) at (0, -5) {$X$};
		\node [style=new style 1] (35) at (3, -5) {$2$};
		\node [style=new style 2] (36) at (3, -2) {$U$};
		\node [style=new style 1, inner sep=1mm, minimum size=0mm] (37) at (6, -2) {$0$};
		\node [style=new style 2] (38) at (6, -5) {$X^{2}$};
		\node [style=none] (39) at (8, -2) {};
		\node [style=none] (40) at (8, -5) {};
		\node [style=none] (41) at (0.25, 4.5) {};
		\node [style=none] (42) at (5, 4.5) {};
		\node [style=none] (43) at (-4, 6.75) {\tiny$\overrightarrow{P}_{1}$};
		\node [style=none] (44) at (-6.75, 4) {\tiny$\overrightarrow{P}_{0}$};
		\node [style=none] (45) at (-1.25, 6) {\tiny $\overrightarrow{P}_{2}$};
		\node [style=none] (46) at (8, 6.75) {\tiny $\overrightarrow{P}_{0}'$};
		\node [style=none] (47) at (10, 6.75) {\tiny$ \overrightarrow{P}_{1}'$};
		\node [style=none] (48) at (12, 6.75) {\tiny$ \overrightarrow{P}_{2}'$};
		\node [style=none] (49) at (1.5, 3.75) {};
		\node [style=none] (50) at (1.5, 2.25) {};
		\node [style=new style 1,inner sep=1mm, minimum size=0mm] (51) at (2.75, 3.75) {$0$};
		\node [style=new style 2] (52) at (2.75, 2.25) {$X$};
		\node [style=none] (53) at (4, 3.75) {};
		\node [style=none] (54) at (4, 2.25) {};
		\node [style=none] (55) at (2.5, 5.25) {\footnotesize Conjugate by};
		\node [style=none] (56) at (3, -7) {c) The final circuit for $U_{[(0,1),(1,2),(2,2)]}$};
		\node [style=none] (57) at (3.5, 0.5) {};
		\node [style=none] (58) at (3, 0.5) {b) Put the three points onto a line of the grid};
		\node [style=none] (59) at (3, 8) {a) The initial 3-level unitary};
		\node [style=none] (60) at (-11, 4) {$U_{[(0,1),(1,2),(2,2)]} \rightarrow$};
		\node [style=none] (61) at (17, 4) {$\leftarrow U_{[(0,2),(1,2),(2,2)]}$};
	\end{pgfonlayer}
	\begin{pgfonlayer}{edgelayer}
		\draw (11) to (14);
		\draw (14) to (17);
		\draw (17) to (18);
		\draw (18) to (19);
		\draw (19) to (16);
		\draw (15) to (16);
		\draw (15) to (18);
		\draw (15) to (14);
		\draw (15) to (12);
		\draw (12) to (13);
		\draw (13) to (16);
		\draw (11) to (12);
		\draw (25) to (28);
		\draw (28) to (29);
		\draw (29) to (30);
		\draw (30) to (27);
		\draw (26) to (27);
		\draw (26) to (29);
		\draw (26) to (25);
		\draw (26) to (23);
		\draw (23) to (24);
		\draw (24) to (27);
		\draw (22) to (23);
		\draw (31.center) to (33);
		\draw (33) to (36);
		\draw (36) to (37);
		\draw (37) to (39.center);
		\draw (37) to (38);
		\draw (36) to (35);
		\draw (33) to (34);
		\draw (32.center) to (34);
		\draw (34) to (35);
		\draw (35) to (38);
		\draw (38) to (40.center);
		\draw (49.center) to (51);
		\draw (51) to (52);
		\draw (50.center) to (52);
		\draw (52) to (54.center);
		\draw (51) to (53.center);
		\draw [style=new edge style 4] (41.center) to (42.center);
		\draw [style=new edge style 0] (25) to (22);
	\end{pgfonlayer}
\end{tikzpicture}
\caption{The equivalent circuit for $U_{[(0,1),(1,2),(2,2)]}$ using controlled gates}
\label{fig:3}

\end{figure}

\begin{figure}
    \centering
        \centering
    $U_{[(0,1),(1,2),(2,0)]}=\begin{blockarray}{(ccc|ccc|ccc)}
    1 & 0 & 0 & 0 & 0 & 0 & 0 & 0 & 0\\
    0 & a_1 & 0 & 0 & 0 & a_2 & a_3 & 0 & 0\\
    0 & 0 & 1 & 0 & 0 & 0 & 0 & 0 & 0\\ \cline{1-9}
    0 & 0 & 0 & 1 & 0 & 0 & 0 & 0 & 0\\
    0 & 0 & 0 & 0 & 1 & 0 & 0 & 0 & 0\\
    0 & b_1 & 0 & 0 & 0 & b_2 & b_3 & 0 & 0\\\cline{1-9}
    0 & c_1 & 0 & 0 & 0 & c_2 & c_3 & 0 & 0\\
    0 & 0 & 0 & 0 & 0 & 0 & 0 & 1 & 0\\
    0 & 0 & 0 & 0 & 0 & 0 & 0 & 0 & 1
  \end{blockarray}$
\begin{tikzpicture}[scale=0.4]
	\begin{pgfonlayer}{nodelayer}
		\node [style=new style 0, inner sep=0mm, minimum size=2mm] (2) at (-20, 5.75) {};
		\node [style=new style 0, inner sep=0mm, minimum size=2mm] (3) at (-18, 5.75) {};
		\node [style=new style 0, inner sep=0mm, minimum size=2mm] (4) at (-16, 5.75) {};
		\node [style=new style 0, inner sep=0mm, minimum size=2mm] (5) at (-20, 3.75) {};
		\node [style=new style 0, inner sep=0mm, minimum size=2mm] (6) at (-18, 3.75) {};
		\node [style=new style 0, inner sep=0mm, minimum size=2mm] (7) at (-16, 3.75) {};
		\node [style=new style 0, inner sep=0mm, minimum size=2mm] (8) at (-20, 1.75) {};
		\node [style=new style 0, inner sep=0mm, minimum size=2mm] (9) at (-18, 1.75) {};
		\node [style=new style 0, inner sep=0mm, minimum size=2mm] (10) at (-16, 1.75) {};
		\node [style=new style 0, inner sep=0mm, minimum size=2mm] (13) at (1, 6) {};
		\node [style=new style 0, inner sep=0mm, minimum size=2mm] (14) at (3, 6) {};
		\node [style=new style 0, inner sep=0mm, minimum size=2mm] (15) at (5, 6) {};
		\node [style=new style 0, inner sep=0mm, minimum size=2mm] (16) at (1, 4) {};
		\node [style=new style 0, inner sep=0mm, minimum size=2mm] (17) at (3, 4) {};
		\node [style=new style 0, inner sep=0mm, minimum size=2mm] (18) at (5, 4) {};
		\node [style=new style 0, inner sep=0mm, minimum size=2mm] (19) at (1, 2) {};
		\node [style=new style 0, inner sep=0mm, minimum size=2mm] (20) at (3, 2) {};
		\node [style=new style 0, inner sep=0mm, minimum size=2mm] (21) at (5, 2) {};
		\node [style=none] (34) at (-18, 6.5) {\tiny $\overrightarrow{P}_{1}$};
		\node [style=none] (35) at (-20.75, 3.75) {\tiny $\overrightarrow{P}_{0}$};
		\node [style=none] (36) at (-16.25, 1) {\tiny $\overrightarrow{P}_{2}$};
		\node [style=none] (37) at (0.25, 6) {\tiny $\overrightarrow{P}_{1}'$};
		\node [style=new style 0, inner sep=0mm, minimum size=2mm] (50) at (-20, -3.75) {};
		\node [style=new style 0, inner sep=0mm, minimum size=2mm] (51) at (-18, -3.75) {};
		\node [style=new style 0, inner sep=0mm, minimum size=2mm] (52) at (-16, -3.75) {};
		\node [style=new style 0, inner sep=0mm, minimum size=2mm] (53) at (-20, -5.75) {};
		\node [style=new style 0, inner sep=0mm, minimum size=2mm] (54) at (-18, -5.75) {};
		\node [style=new style 0, inner sep=0mm, minimum size=2mm] (55) at (-16, -5.75) {};
		\node [style=new style 0, inner sep=0mm, minimum size=2mm] (56) at (-20, -7.75) {};
		\node [style=new style 0, inner sep=0mm, minimum size=2mm] (57) at (-18, -7.75) {};
		\node [style=new style 0, inner sep=0mm, minimum size=2mm] (58) at (-16, -7.75) {};
		\node [style=new style 2] (62) at (-16.75, -11.25) {$X^{2}$};
		\node [style=new style 2] (63) at (-13.75, -11.25) {$X$};
		\node [style=new style 2] (64) at (-1.75, -11.25) {$X^{2}$};
		\node [style=new style 2] (65) at (1.25, -11.25) {$X$};
		\node [style=new style 1, inner sep=0.5mm] (67) at (-10.75, -11.25) {$0$};
		\node [style=new style 1,inner sep=0.5mm] (68) at (-7.75, -11.25) {$0$};
		\node [style=new style 1, inner sep=0.5mm] (69) at (-16.75, -14.25) {$2$};
		\node [style=new style 1, inner sep=0.5mm] (70) at (-13.75, -14.25) {$0$};
		\node [style=new style 2] (72) at (-10.75, -14.25) {$X^2$};
		\node [style=new style 2] (73) at (-7.75, -14.25) {$U$};
		\node [style=new style 1, inner sep=0.5mm] (74) at (-1.75, -14.25) {$0$};
		\node [style=none] (75) at (-18.75, -11.25) {};
		\node [style=none] (76) at (-18.75, -14.25) {};
		\node [style=new style 1, inner sep=0.5mm] (77) at (1.25, -14.25) {$2$};
		\node [style=none] (78) at (3.25, -11.25) {};
		\node [style=none] (79) at (3.25, -14.25) {};
		\node [style=none] (80) at (-21, -3.75) {};
		\node [style=none] (81) at (-21, -7.75) {};
		\node [style=none] (82) at (-20.5, -3.75) {};
		\node [style=none] (83) at (-20.5, -5.25) {};
		\node [style=none] (84) at (-20.5, -6.25) {};
		\node [style=none] (85) at (-20.5, -7.75) {};
		\node [style=none] (86) at (0.25, 2.25) {\tiny $\overrightarrow{P}_{2}'$};
		\node [style=none] (87) at (-13.25, 4.25) {};
		\node [style=none] (89) at (-2, 4.25) {};
		\node [style=none] (90) at (-8, 5) {\footnotesize {conjugate by}};
		\node [style=none] (91) at (-12.25, 3.25) {};
		\node [style=none] (92) at (-12.25, 1.5) {};
		\node [style=none] (93) at (-9.75, 3.25) {};
		\node [style=none] (94) at (-9.75, 1.5) {};
		\node [style=new style 2, inner sep=1mm] (95) at (-11, 3.25) {$X^{2}$};
		\node [style=new style 1, inner sep=0.5mm, minimum size=0mm] (96) at (-11, 1.5) {$2$};
		\node [style=none] (97) at (-6.25, 3.25) {};
		\node [style=none] (98) at (-6.25, 1.5) {};
		\node [style=none] (99) at (-3.75, 3.25) {};
		\node [style=none] (100) at (-3.75, 1.5) {};
		\node [style=new style 2] (101) at (-5, 3.25) {$X$};
		\node [style=new style 1, inner sep=0.5mm, minimum size=0mm] (102) at (-5, 1.5) {$0$};
		\node [style=none] (103) at (-8, 2.5) {\footnotesize {and}};
		\node [style=none] (104) at (-9.25, -7.25) {};
		\node [style=none] (105) at (-9.25, -5.5) {};
		\node [style=none] (106) at (-6.75, -7.25) {};
		\node [style=none] (107) at (-6.75, -5.5) {};
		\node [style=new style 2, inner sep=1mm] (108) at (-8, -7.25) {$X^{2}$};
		\node [style=new style 1, inner sep=0.5mm, minimum size=0mm] (109) at (-8, -5.5) {$0$};
		\node [style=none] (110) at (-11, -4.5) {};
		\node [style=none] (111) at (-4.75, -4.5) {};
		\node [style=none] (112) at (-8, -3.75) {\footnotesize {conjugate by} };
		\node [style=none] (115) at (0.25, 4.25) {\tiny $\overrightarrow{P}_{0}'$};
		\node [style=none] (117) at (-9, -0.75) {};
		\node [style=none] (118) at (-9, -0.75) {};
		\node [style=none] (119) at (-8, -0.75) {b) First put the three points onto a line of the grid};
		\node [style=none] (120) at (-8, -9) {c) Permute the rows to the correct order};
		\node [style=none] (121) at (-8, -16) {d) Final circuit for $U_{[(0,1),(1,2),(2,0)]}$};
		\node [style=none] (122) at (-8, 8.5) {a) The initial 3-level unitary};
		\node [style=none] (123) at (-20, -3) {\tiny $\overrightarrow{P}_{1}'$};
		\node [style=none] (124) at (-20, -8.5) {\tiny $\overrightarrow{P}_{2}'$};
		\node [style=none] (125) at (-20.75, -5.75) {\tiny $\overrightarrow{P}_{0}'$};
		\node [style=none] (126) at (9, 4) {$\leftarrow U_{[(0,1),(0,2),(0,0)]}$};
		\node [style=none] (127) at (-25, 3.75) {$U_{[(0,1),(1,2),(2,0)]} \rightarrow$};
		\node [style=none] (128) at (8, -5.25) {$\leftarrow U_{[(0,0),(0,1),(0,2)]}$};
		\node [style=new style 0, inner sep=0mm, minimum size=2mm] (129) at (-1, -3.25) {};
		\node [style=new style 0, inner sep=0mm, minimum size=2mm] (130) at (1, -3.25) {};
		\node [style=new style 0, inner sep=0mm, minimum size=2mm] (131) at (3, -3.25) {};
		\node [style=new style 0, inner sep=0mm, minimum size=2mm] (132) at (-1, -5.25) {};
		\node [style=new style 0, inner sep=0mm, minimum size=2mm] (133) at (1, -5.25) {};
		\node [style=new style 0, inner sep=0mm, minimum size=2mm] (134) at (3, -5.25) {};
		\node [style=new style 0, inner sep=0mm, minimum size=2mm] (135) at (-1, -7.25) {};
		\node [style=new style 0, inner sep=0mm, minimum size=2mm] (136) at (1, -7.25) {};
		\node [style=new style 0, inner sep=0mm, minimum size=2mm] (137) at (3, -7.25) {};
		\node [style=none] (138) at (-1.75, -3.25) {\tiny $\overrightarrow{Q}_2$};
		\node [style=none] (139) at (-1.75, -7) {\tiny $\overrightarrow{Q}_0$};
		\node [style=none] (140) at (-1.75, -5) {\tiny $\overrightarrow{Q}_1$};
		\node [style=new style 1, inner sep=0.5mm] (141) at (-4.75, -11.25) {0};
		\node [style=new style 2] (142) at (-4.75, -14.25) {$X$};
	\end{pgfonlayer}
	\begin{pgfonlayer}{edgelayer}
		\draw (2) to (5);
		\draw (5) to (8);
		\draw (8) to (9);
		\draw (9) to (10);
		\draw (10) to (7);
		\draw (6) to (7);
		\draw (6) to (9);
		\draw (6) to (5);
		\draw (6) to (3);
		\draw (3) to (4);
		\draw (4) to (7);
		\draw (2) to (3);
		\draw (16) to (19);
		\draw (21) to (18);
		\draw (17) to (18);
		\draw (17) to (20);
		\draw (17) to (16);
		\draw (17) to (14);
		\draw (14) to (15);
		\draw (15) to (18);
		\draw (13) to (16);
		\draw (50) to (53);
		\draw (53) to (56);
		\draw (56) to (57);
		\draw (57) to (58);
		\draw (58) to (55);
		\draw (54) to (55);
		\draw (54) to (57);
		\draw (54) to (53);
		\draw (54) to (51);
		\draw (51) to (52);
		\draw (52) to (55);
		\draw (50) to (51);
		\draw (75.center) to (62);
		\draw (62) to (63);
		\draw (67) to (68);
		\draw (68) to (64);
		\draw (64) to (65);
		\draw (64) to (74);
		\draw (68) to (73);
		\draw (67) to (72);
		\draw (63) to (70);
		\draw (62) to (69);
		\draw (76.center) to (69);
		\draw (69) to (70);
		\draw (72) to (73);
		\draw (73) to (74);
		\draw (74) to (77);
		\draw (65) to (77);
		\draw (77) to (79.center);
		\draw (65) to (78.center);
		\draw (91.center) to (95);
		\draw (95) to (93.center);
		\draw (92.center) to (96);
		\draw (96) to (94.center);
		\draw (95) to (96);
		\draw (97.center) to (101);
		\draw (101) to (99.center);
		\draw (101) to (102);
		\draw (98.center) to (102);
		\draw (102) to (100.center);
		\draw (104.center) to (108);
		\draw (108) to (106.center);
		\draw (105.center) to (109);
		\draw (109) to (107.center);
		\draw (108) to (109);
		\draw (63) to (67);
		\draw (70) to (72);
		\draw (132) to (135);
		\draw (137) to (134);
		\draw (133) to (134);
		\draw (133) to (136);
		\draw (133) to (132);
		\draw (133) to (130);
		\draw (130) to (131);
		\draw (131) to (134);
		\draw (129) to (132);
		\draw (137) to (135);
		\draw (130) to (129);
		\draw (141) to (142);
		\draw [style=new edge style 2, bend left=15] (81.center) to (80.center);
		\draw [style=new edge style 0, bend right=15, looseness=1.25] (82.center) to (83.center);
		\draw [style=new edge style 0, bend right=15] (84.center) to (85.center);
		\draw [style=new edge style 4] (87.center) to (89.center);
		\draw [style=new edge style 4] (110.center) to (111.center);
		\draw [style=new edge style 2] (21) to (20);
		\draw [style=new edge style 2] (20) to (19);
		\draw [style=new edge style 2] (14) to (13);
	\end{pgfonlayer}
\end{tikzpicture}

\caption{The equivalent circuit for $U_{[(0,1),(1,2),(2,0)]}$ using controlled gates}
\label{fig:4}
\end{figure}



\subsection{Proof of Theorem \ref{generalizationGS}}

\begin{proof}[Proof of Theorem \ref{generalizationGS}]
  
    When $n=1$, we have $\U_{3^n}(\Z[\frac{1}{3},\omega]) = $ Clifford$+R$ \cite{kalra2023synthesis} and so the statement of the theorem for $n=1$ is essentially the same as \cite[Corollary 23]{metaplectic}.\\
    
    When $n\ge 2$, using Lemma \ref{level_matrices_lemma} we see that $U \in \U_{3^n}(\Z[\frac{1}{3},\omega])$ can be decomposed into 3-level unitaries of type $X, S$ and $H$; 2-level unitaries of type $\begin{pmatrix}
        0 & 1\\ 1 & 0
    \end{pmatrix} $; and 1-level unitaries of type $\omega$ and -1. Since $3^n \ge 3^2 > 3$, so a 1-level unitary of type -1 or $\omega$ is equivalent to a 3-level unitary of type $X_{\sigma}^{-1}RX_{\sigma}$ or $X_{\sigma}^{-1}SX_{\sigma}$ respectively for $\sigma = (1\;2\;3)$ or $(1\;3\;2)$. Further, a two-level unitary of type $\begin{pmatrix}
        0 & 1\\ 1 & 0
    \end{pmatrix}$ is equivalent to a 3-level unitary of type $H^2$. Therefore, $U$ can be decomposed into a sequence of 3-level unitaries of type $X,S,H$ and $R$.\\
    
    Using Lemma \ref{convert_controlled_lemma}, we see that any such 3-level untaries can be decomposed into $\ket{2}^{\otimes (n-1)}$-controlled gates of type $X,S,H$ and $R$. Note that $X, S$ and $H$ gates are single-qutrit Clifford gates. Using \cite[Theorem 1]{Yeh_2022}, all these controlled gates can be decomposed into a $n$-qutrit Clifford$+T$ circuit.\\

    Further, using \cite{metaplectic}, the gate $R$ can be viewed as an element of 2-qutrit Clifford$+T$. Therefore using \cite[Theorem 1]{Yeh_2022}, we see that $\ket{2}^{\otimes (n-1)}$-controlled-$R$ can be decomposed into a $(n+1)$-qutrit Clifford$+T$ circuit.\\

    Hence we see that $U\in U_{3^n}(\Z[\frac{1}{3},\omega])$ can be decomposed into a circuit consisting of qutrit Clifford$+T$ gates with at most one ancilla (when $R$ can is involved).  
\end{proof}

\section{Unitaries with entries from $\mathbb{Z}{ [\frac{1}{3},\zeta]}$}\label{sec4}

In this section we obtain an algorithm to represent an arbitrary $3^n \times 3^n$ unitary matrix with entries in $\mathbb{Z}{ [\frac{1}{3},\zeta]}$ in terms of a $(n+2)-$ qutrit Clifford $+ T$
circuit. We do this by composing the catalytic embedding 
$\Phi :  \U_{3^n}(\Z[\frac{1}{3},\zeta]) \hookrightarrow \U_{3^{n+1}}(\Z[\frac{1}{3},\omega])$
with the embedding $\Psi : \U_{3^{n+1}}(\Z[\frac{1}{3},\omega]) \hookrightarrow (C + T)^{n+2}$
obtained in Theorem \ref{th1} to get an embedding 

$$\U_{3^n}(\Z[\frac{1}{3},\zeta]) \hookrightarrow (C + T)^{n+2}.$$

We already know that single-qutrit Clifford$+T$ gates are not enough to decompose single-qutrit gates with entries from $\Z[\frac{1}{3},\zeta]$. Theorem 2 gives an algorithm to implement such gates using a Clifford $+ T$ circuit given we have two more ancillae.
In particular, for $n=1$, our theorem takes a single-qutrit Clifford$+\mathcal{D}$ unitary and embed it inside 3-qutrit unitaries and decomposes this embedded unitary into 3-qutrit-Clifford$+T$. In general, an $n$-qutrit gate with entries from $\Z[\frac{1}{3},\zeta]$ is represented in $(n+2)$-qutrit Clifford$+T$.
At first we embed $\U_{3^n}(\Z[\frac{1}{3},\zeta])$ into $\U_{3^{n+1}}(\Z[\frac{1}{3},\omega])$ using catalytic embedding which adds one ancilla.
The final step is to apply the algorithm from Theorem \ref{th1}.

The following theorem is a complete analogue for qutrits of the main result in 
\cite{giles2013exact} with the ring $\Z[\frac{1}{\sqrt{2}},i]$ replaced by the natural ring $\mathbb{Z}{ [\frac{1}{3},\zeta]}$ for multi-qutrit Clifford$+T$. 

\begin{theorem}\label{th2}
Let $U$ be a unitary $3^n \times 3^n$ matrix. Then the following are equivalent:
\begin{enumerate}[label=(\alph*)]
\item $U$ can be exactly represented by a quantum circuit over Clifford$+T$ gate set, possibly using some finite number of ancillae.
\item  The entries of $U$ belong to the ring $\Z[\frac{1}{3},\zeta]$
\end{enumerate}
Further, in (a), at most two ancillas are sufficient.
\end{theorem}

\begin{proof}
The implication (a) $\implies$ (b) is trivial. For the converse implication, we first use the catalytic embedding described in Section \ref{seccat} to embed $\U_{3^n}(\Z[\frac{1}{3},\zeta])$ inside $\U_{3^{n+1}}(\Z[\frac{1}{3},\omega])$. This step adds one ancilla (the catalyst involved in the embedding becomes the ancilla). Then we use Theorem \ref{generalizationGS} to construct a circuit over Clifford$+T$ with at most one more ancilla. This concludes the theorem.
\end{proof}
\section{Conclusion}
In this paper, we extended the celebrated Giles and Selinger's result for multi-qubit Clifford$+T$ circuits to the multi-qutrit case. We presented an algorithm to perform synthesis for unitaries over $\mathcal{U}_{3^n}(\mathbb{Z}[\frac{1}{3},e^{2\pi i/3}])$ using the Clifford$+T$ gate with at most a single extra ancilla. We further combined this result with the recently proposed catalytic embedding method to give an algorithm for the exact synthesis of a unitary in $\mathcal{U}_{3^n}(\mathbb{Z}[\frac{1}{3},e^{2\pi i/9}])$ over the multi-qutrit Clifford$+T$ gate set.

Ideally one would like to extend these results for $p> 3$ but this turns out to be much harder because the sde profile (amounts to a building structure) for $p>3$ is significantly more complicated. In particular, no exact synthesis algorithm seems to be known for single qudit gates of dimension $p > 3$.